\documentclass[a4paper,12pt]{amsart}
\usepackage{graphicx}
\usepackage[]{color}

\usepackage[]{tikz} \usetikzlibrary{matrix, chains, arrows}

\usepackage[margin=3cm]{geometry}
\newtheorem{thm}{Theorem} 

\newtheorem{lem}[thm]{Lemma}
\newtheorem{prop}[thm]{Proposition}
\theoremstyle{definition}

\theoremstyle{remark}

\numberwithin{equation}{section}
\newcommand{\PP}{\mathcal{P}}

\title{
Can we `future-proof'  consensus trees?}%
\author{David Bryant, Andrew Francis, Mike Steel}
\thanks{}%
\date{\today}

\begin{document}

\begin{abstract}
Consensus methods are widely used for combining phylogenetic trees into a single estimate of the evolutionary tree for a 
group of species.  As more taxa are added, the new source trees may begin to tell a different evolutionary story when restricted to the original set of taxa. However, if the new trees, restricted to the original set of taxa, were to agree exactly with the earlier trees, then we might hope that their consensus would either agree with or resolve the original consensus tree.  In this paper, we ask under what conditions consensus methods exist that are `future proof' in this sense.  While we show that some methods (e.g. Adams consensus) have this property for specific types of input, we  also establish a rather surprising `no-go' theorem: there is no `reasonable' consensus method that satisfies the future-proofing property in general. We then investigate a second  notion of  `future proofing' for consensus methods, in which trees (rather than taxa) are added, and establish some positive and negative results.  We end with some questions for future work.  
\end{abstract}

\maketitle


\section{Introduction}

Consensus methods are widely used to combine phylogenetic trees into a single phylogeny. Traditionally, consensus methods have been applied as a technique for combining information from diverse data sources 
\cite{swo91}. More recently, consensus methods are most commonly used as a technique for summarizing output from MCMC analyses (e.g. \cite{10hol08}), or estimating species trees from gene trees (e.g. \cite{7deg09}).  

In these situations, consensus methods provide a way to combine phylogenetic trees into a single tree representing the underlying evolutionary history of the species under study.

Suppose we have a collection of trees for a set $S$ of taxa that we have combined by some consensus method to produce an estimate $T$ of the tree for $S$.   In future, we may collect data on additional species, build trees from these data, and infer a consensus tree $T'$ from these trees. Of course, the  tree inferred for our original set of species $S$ may well differ from what we obtain today. 
But suppose that the trees they use to build their consensus tree agree exactly with the trees we
have today once we restrict attention to relationships between the species in $S$. In that case, we might hope that the tree we obtain by restricting $T'$ to $S$ should also agree with our existing tree $T$, or at least
be a refinement of $T$ (perhaps because the additional taxa help us to resolve polytomies that cannot be resolved at present).

In other words, if the future phylogenies agree with our present ones when restricted to the taxa we have available today, can we hope that the evolutionary relationships established today using consensus are `safe' in the sense that different branching orders need not be postulated in future?

We call this property of a consensus method {\em extension stability}. It asserts that if a consensus tree has been produced from a profile of trees, and a new species is added to each of them, one expects that the new consensus tree on the profile of extended trees will restrict back down to the consensus of the original profile (or at least display the same monophyletic clusters). A formal definition is given in the next section, and Fig.~\ref{f:motivation} illustrates this concept for a particular consensus method (majority rule) applied to a particular input profile of three binary trees. In this example, an additional taxon $x$ appears in different places in the input trees, however the consensus tree on the full taxon set (including $x$) still displays
the original consensus relationship between the original taxa $a,b,c$. For other profiles, majority rule consensus can fail to satisfy extension stability.

\begin{figure}[htb]
\center
\includegraphics[width=10cm]{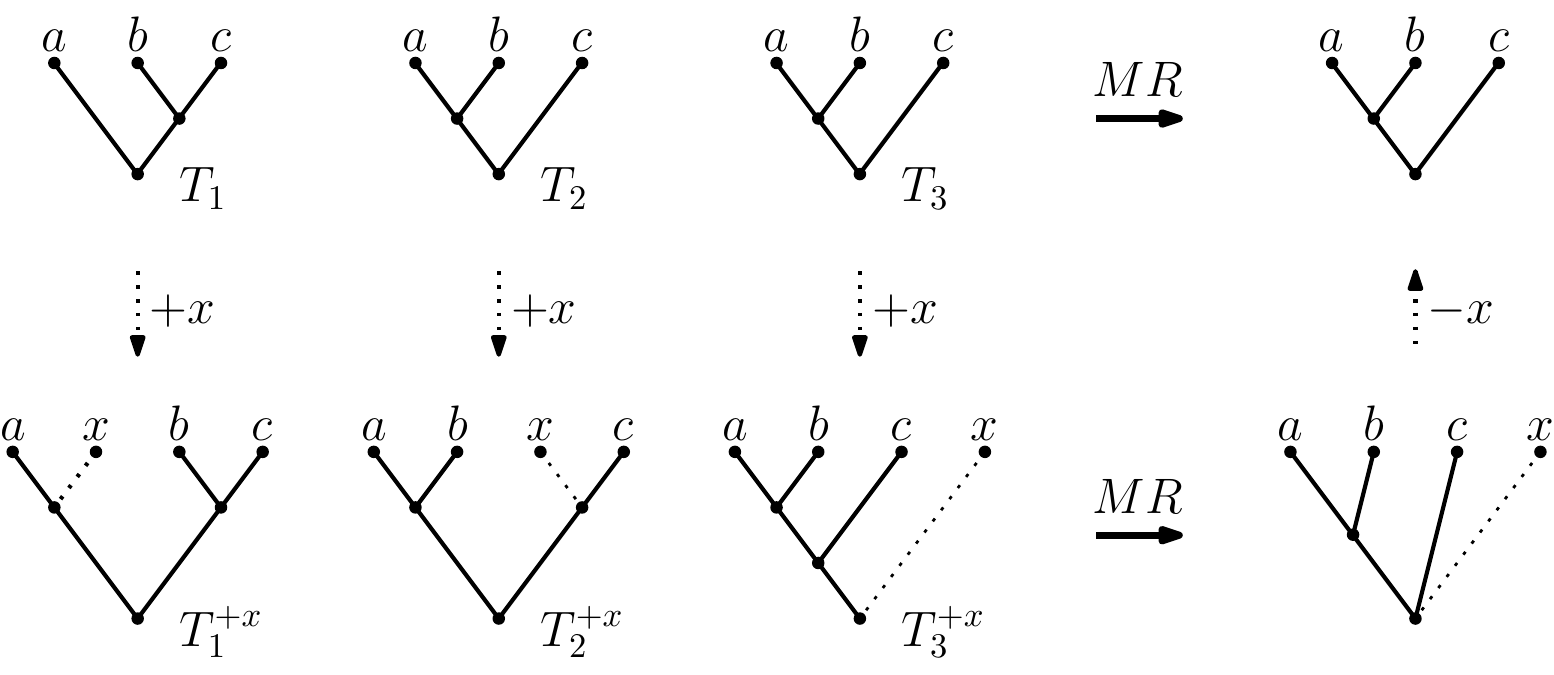}
\caption{An example where extension stability applies. The majority rule (MR) consensus of the trees $T_1, T_2$ and $T_3$ is the resolved tree $ab|c$.  If an additional taxon $x$
is analysed and the resulting trees $T_1^x, T_2^x, T_3^x$ agree with $T_1, T_2$ and $T_3$ up to the placement  of $x$, then the majority rule consensus of the trees still
displays the original relationship $ab|c$ once taxon $x$ is removed.}\label{f:motivation}
\end{figure}


In this paper, we explore a very simple question:  which, if any, reasonable consensus methods can achieve this goal?
Surprisingly, there are standard methods (including Adams consensus) that do not satisfy extension stability \cite{15vel}.  The question of whether `reasonable' extension stable consensus functions exist at all was left open in that paper. 
Here we show that remarkably, in general, no such consensus function exists.  Our `no-go' theorem is of a similar spirit but quite different from (and deeper than) the non-existence result of Proposition 2 from \cite{16ste00}.  We do, however, show that under certain restrictions, either on the profile of input trees or by modifying the extension stability condition, there exist consensus methods (including Adams consensus) that are stable in this limited sense.

In a later section, we also consider a further type of stability condition that arises when trees (rather than taxa) are added.  This condition means that one can simply update the existing consensus tree when a new input tree is sampled, and the result will be the same as re-analysing the entire sequence of input trees. We show that certain methods have this property but that known methods that preserve nesting or triplet relationships do not.

Our approach provides further applications of the axiomatic approach to consensus methods in systematic biology pioneered by
F.R. McMorris and colleagues over more than three decades (for an overview, see \cite{6day03}, with more recent results in \cite{8don11}, \cite{13mcm}).   This has led to a range of theorems which demonstrate that no consensus method can simultaneously satisfy various combinations of criteria, each of which seems both desirable and reasonable. In some cases, the results are particular to unrooted trees (for example, \cite{12mcm85, 16ste00}) while in other cases they apply in the rooted setting.   
  
This axiomatic approach is motivated, in part, by the celebrated theorem of Kenneth Arrow \cite{3arr} who showed that four seemingly very plausible criteria cannot be simultaneously achieved by any voting method. While these `impossibility' results (in voting theory and in phylogenetics) may be seen as negative, they have the positive effect of focusing attention on which criteria should be dropped, or at least weakened, to devise  methods that achieve various combinations of desirable attributes.

\bigskip

\mbox{}

\bigskip

\section{Basic notation}
\noindent {\bf Trees and clusters} 

In this paper, we deal exclusively with rooted phylogenetic trees and (following \cite{14sem03}), we will let $RP(X)$ denote the
set of rooted phylogenetic trees on leaf set $X$ of taxa.  
 A {\em cluster} of a tree is a subset of the leaves that forms a monophyletic clade (i.e. it is the set of all leaves that are descendants  of some vertex).  It is well known that the set of clusters of any rooted phylogenetic tree forms a hierarchy (i.e. two clusters are either disjoint or one is a subset of the  other) and there is a one-to-one correspondence between the trees in $RP(X)$ (up to equivalence)
 and the set of hierarchies on $X$ that contain $X$ and the singleton elements (and not containing the empty set). 
 
 A tree $T' \in RP(X)$ {\em refines} a tree $T \in RP(X)$ if every cluster of $T$ is present in $T'$;  we denote this by writing
 $T \preceq T'$. Thus $RP(X)$ forms a partially ordered set (poset) with the star tree as the unique minimal element and the binary (fully resolved) trees as the maximal elements
 ($RP(X)$ is also a median semilattice, which is a concept that has been developed in more general consensus settings \cite{mcm95}). 
  Given a tree $T \in RP(X)$ and a subset $Y$ of $X$,  we let $T|_Y \in RP(Y)$ be the tree obtained from $T$ by restricting $T$ to just the leaves in $Y$.  When $|Y|=3$ and $T|_Y$ is binary, we say that this tree is a {\em rooted triple} and is {\em displayed} by $T$ (we write $ab|c$ if this rooted triple has $a,b$ as a cherry and $c$ as the third leaf). 
 
 \bigskip
 
\noindent {\bf Consensus methods} 
 
 A {\em profile} of trees  from $RP(X)$  is an ordered tuple of trees
$(T_1, \ldots,T_k) \in RP(X)^k$ for some $k$; we call these trees the `input trees'.  Notice that the same tree can appear more than once. A (phylogenetic) {\em consensus method} is a function $\varphi$ that, for every set $Y$ of taxa and every number $k\geq 1$, associates to each profile of $k$ trees from $RP(Y)$  a corresponding tree in $RP(Y)$. 

The most widely used consensus method in contemporary phylogenetics is probably {\em majority rule}, which returns the tree that has clusters present in more than half of the input trees. Variations on this include {\em strict consensus}, which returns the tree that has clusters present in every input tree, and {\em loose consensus}, which returns the tree that has clusters that are (i) present in at least one input tree and (ii) are compatible with all other input trees.  More recent variations on these consensus methods have also been devised, including the majority (+) method \cite{8don11} and frequency-difference consensus \cite{11jan}.  All these methods satisfy the (Pareto) property that if a cluster is present in all the input trees, then that cluster will also be present in the output.

Another type of consensus method outputs the tree that occurs most often in the profile (in the case of ties, then some consensus of these tied trees ---  perhaps using majority rule or strict consensus --- can be applied).  Such methods are particularly relevant when summarising a posterior distribution on trees from a sample generated by MCMC.

A quite different class of consensus methods  aims to also preserve  other features shared across the input trees, such as nestings or 3-taxon relationships. We describe two such methods.  One is the well-known {\em Adams consensus} method  \cite{1ada86}, which preserves nestings, and which we denote by $\varphi_{Ad}$.  The other is a lesser known consensus method, called  {\em local consensus} in \cite{5bry03}, and is based on applying the BUILD algorithm of  \cite{2aho81} to the set of rooted triples shared by all the input trees in the profile; we will here refer to this method as  {\em Aho consensus}, denoted by $\varphi_{Ah}$. 

Both Adams and Aho consensus  associate to each profile  $\PP = (T_1, T_2, \ldots, T_k) \in RP(X)^k$ a partition $\Pi(\PP)$ of $X$ that forms the maximal clusters of the consensus tree. 
For Adams consensus, $\Pi(\PP)$  equals the nonempty intersections of the maximal clusters of the trees in $\PP$.
For Aho consensus, $\Pi(\PP)$ is the connected components of the graph $(X, E_\PP)$, where $E_\PP$ is the set of pairs $a,b \in X$ for which there is an element $c \in X$ with
$ab|c$ displayed by every tree in $\PP$ (we will refer to this as the {\em BUILD graph} for $\PP$). Once $\Pi(\PP)$ has been determined, each method then repeats the same process on the restriction of $\PP$ to each set $A \in \Pi(\PP)$, 
eventually producing a hierarchy of sets on $X$ and thus a well-defined consensus tree in $RP(X)$ (for further details, see \cite{5bry03, 6day03, 14sem03}). 
 
Both Adams and Aho consensus satisfy the following property of being {\em Pareto on rooted triples}, which is the following condition:
\begin{equation}
\label{trip1}
\mbox{ If each tree  $T \in \PP$ displays  $ab|c$, then so does $\varphi(T)$.}
\end{equation}
This condition is equivalent to the following property (for all $Y \subseteq  X$):
\begin{equation}
\label{trip2}
\mbox{If $T' \preceq T|_Y$ for every $T \in \PP$  then $T' \preceq \varphi(T)|_Y$.}
\end{equation}
 
Note that any consensus method that is Pareto on rooted triples is also Pareto on clusters.
Adams consensus also has the property that if $ab|c$ is displayed by $\varphi_{Ad}(\PP)$, then at least one tree in $\PP$ must display $ab|c$ \cite{5bry03}. 

 All the consensus methods  we have discussed (strict consensus, loose consensus, majority rule,  Adams  and Aho consensus) also 
satisfy the following property: given a profile $\PP$ and its consensus tree, if one adds the consensus tree to the profile, then the consensus on the enlarged profile is unchanged.    Formally, this is the condition that:
\begin{equation}
\label{stab}
\varphi(\PP\cdot \varphi(\PP))=\varphi(\PP),
\end{equation}
where $\cdot$ refers to the concatenation of the sequence $\PP$ with the consensus tree $\varphi(\PP)$.
This is analogous to the behaviour of the averages of real numbers (indeed, $y$ is the average of $x_1, \ldots, x_n$ if and only if $y = {\rm av}(x_1, x_2, \ldots, x_n,y)$).  A proof that Eqn.~(\ref{stab}) holds for the main consensus methods mentioned so far is provided in the Appendix.  Condition (\ref{stab}) corresponds to an ``Invariance'' property in \cite{mcm12}.

\section*{Regular properties for consensus methods}\label{s:regular}

To move beyond  the study of particular consensus methods, we impose only very minimal requirements on a possible consensus method. 
For this paper, we will call a consensus method $\varphi$ ``regular'' if it satisfies the following three minimal axioms:

\begin{enumerate}
	\item \textbf{Unanimity}. If the trees in $\PP$ are all the same tree $T$, then $\varphi(\PP)=T$. 
	\item \textbf{Anonymity}. Changing the order of the trees in $\PP$ does not change $\varphi(\PP)$.
	\item \textbf{Neutrality}. Changing the labels on the leaves of the trees in $\PP$ simply relabels the leaves of $\varphi(\PP)$ in the same way. 
\end{enumerate}

Anonymity corresponds to a mathematical condition referred to as commutativity.  It is the condition that for any profile $\PP = (T_1, T_2, \ldots, T_k)$ and any permutation $\rho$ on $\{1, \ldots, k\}$, we have:
$\varphi(\PP) = \varphi(\rho(\PP))$, where $\rho(\PP)$ is the profile $(T_{\rho(1)}, \ldots, T_{\rho(k)}).$
Nearly all existing consensus methods satisfy this condition of not  taking the order of the input trees into account.  Moreover, even for consensus methods that do allow trees to be ranked according to the strength of support, such methods should still be able to deal with inputs in which the trees are to be regarded of equal value, and in that case unanimity should hold.

Neutrality \cite{6day03} corresponds to a mathematical condition referred to as equivariance.  Informally speaking, it is the requirement that the names given to the objects labelling the leaves of the tree should not play any special role in deciding how the consensus tree is determined. For example, if `dog' and `cat' were interchanged in each input tree, then the output tree should just be the same consensus tree but with `dog' and `cat' interchanged.

More precisely, given a  bijection $\alpha: X \rightarrow X'$, and any tree $T \in RP(X)$, let $T^\alpha$ be the tree in $RP(X')$ 
obtained from $T$ by replacing leaf $x$ with leaf $\alpha(x)$ for each $x \in X$.   Then neutrality is the condition that for any profile 
$\PP = (T_1, T_2, \ldots, T_k)$
and any bijection $\alpha: X \rightarrow X'$, we have $\varphi(\PP^\alpha) = \varphi(\PP)^\alpha$, where $\PP^\sigma = (T^\alpha_1, T^\alpha_2, \ldots, T^\alpha_k)$.


We will refer to any consensus method that satisfies these three properties -- unanimity, anonymity and neutrality  -- as a {\em regular} consensus method.  All standard phylogenetic consensus methods (e.g. strict consensus, majority rule, loose consensus and Adams consensus) satisfy these properties.

The additional property that is the main focus of this paper is \emph{extension stability}, which is defined as follows.  
\begin{enumerate}\setcounter{enumi}{3}
	\item \textbf{Extension stability}. For every phylogenetic profile $\PP$ consisting of trees from a leaf set $X$ of arbitrarily large size, and  any nonempty set $Y\subseteq X$, $$\varphi(\PP|_Y)\preceq \varphi(\PP)|_Y.$$
\end{enumerate}

Extension stability was defined in \cite{15vel}, though it was referred to there as a ``weak independence'"condition. It is different from any of the eight ``stability conditions" $(I_1)$--$(I_8)$ defined and studied by \cite{4bar3}, and weaker than the closest analogue in that list, namely $(I_6)$.  Indeed, any consensus method that satisfies unanimity and extension stability is also Pareto on rooted triples (Eqns.~\ref{trip1}, \ref{trip2}).

Strict consensus satisfies the reverse inequality described by extension stability, namely: 
\begin{equation}
\label{reverse} 
\varphi(\PP)|_Y \preceq \varphi(\PP|_Y)
\end{equation}
(the short proof is given in the Appendix).
The fact that this inequality can be strict is well known. For example, if $X=Y \cup \{x\}$, where $x$ is a `rogue' taxon (external to $Y$) which positions itself in different places across the trees in $\PP$, then the strict consensus of the trees in  the resulting profile can be largely unresolved (and so remains unresolved when restricted to $Y$)  compared to the strict consensus of the trees
on the pruned leaf set $Y$. An example of this phenomenon arises  in the study of whale phylogeny in \cite{9gei09}, where the inclusion of the extinct taxon {\em Andrewsarchus} causes the strict consensus tree to be much less resolved than if it is removed (as shown in Appendix Fig. 2 of that paper).   
This reverse form of extension stability (Eqn.~\eqref{reverse}) fails to capture, however, the requirement that information (clusters) determined today is `safe' when future trees
agree with the ones of today on the taxa subset $Y$. To make this more obvious, observe that another regular consensus method satisfying  Eqn.~\eqref{reverse} is the  rather trivial consensus method $\varphi$ for which $\varphi(\PP)$ is always the star tree (on the same leaf set as $\PP$), except when $\PP$ is of the form $(T, T, \ldots, T)$, in which case
$\varphi(\PP) = T$.

Note that there are regular consensus methods that satisfy a restriction of the extension stability property where $X$ is required to be sufficiently small. For example, strict consensus satisfies these properties for $|X|\leq 3$, while for $|X|=4$, Adams consensus satisfies  the extension property (however, that method does not satisfy this property for $|X|=5$ as Fig.~\ref{f:adams.binary} shows).  Notice that extension stability requires that any consensus choice we make on a small leaf set will still apply for profiles of larger trees when we restrict these down to smaller leaf sets; we will see that this provides additional constraints, since the consensus trees on the sub-problems need to at least be compatible (consistent with {\em some} tree).

There are also regular consensus methods (e.g. Adams consensus) that satisfy a restriction of the extension stability property where $Y$ is required to have  size at most 3 (regardless of the size of $|X|$).   
Also, if we were to drop any one of the conditions required for a consensus method to be regular, then, as noted in \cite{15vel},
one can readily construct a method that is extension stable and satisfies the other regularity conditions. For example, if we drop anonymity, we can simply take the consensus of $(T_1, T_2, \ldots, T_k)$ to be $T_1$. 

The observation that Adams consensus (a standard and regular consensus method)  fails to satisfy this extension stability --- thanks to an example by R. Powers, reported in \cite{15vel} ---  is the starting point of this study (a simpler example involving binary trees that shows  that Adams consensus  fails is given in Figure~\ref{f:adams.binary}). The question of whether {\em any} regular consensus method can satisfy extension stability, left open by \cite{15vel},  is a main focus of this paper. 
We begin with some positive results.


\section{Positive results for extension stability}
\label{posec}

In this section, we describe some particular settings where extension stability provably holds for certain consensus methods.

\subsection*{Regular consensus methods that are extension stable for short trees.}

The {\em height} of a rooted phylogenetic tree is the length of the longest path from the root to a leaf.  We now show that for profiles of trees of height at most two, which we call {\em level 2} trees,  there exist regular consensus functions that are extension stable. Two such methods are Adams and Aho consensus.  Note that these are different consensus methods on level 2 trees, since for the profile
 $\PP = (((abcd), (ef)), ((cdef), (ab)), ((abef), (cd)))$,
$\varphi_{Ad}(\PP) = ((ab), (cd), (ef))$ whereas $\varphi_{Ah}(\PP)$ is the star tree.

\begin{thm}
\label{posthm}
If all trees in the profile $\PP$ have height at most 2, then Adams consensus $\varphi_{Ad}$ and Aho consensus $\varphi_{Ah}$
are both extension stable.
\end{thm}

\begin{proof}
Suppose that  every tree in $\PP=(T_1, \ldots, T_k)$ (a  profile of trees each having leaf set $X$) has height at most 2.
Then for any subset $Y$ of $X$, $\PP|_Y$ also has this property (every tree in $\PP|_Y$
has height at most 2). 
To prove the result it is sufficient to show that every rooted triple in $\varphi_{Ad}(\PP|_Y)$ is also displayed by $\varphi_{Ad}(\PP)$.

If $ab|c$ is a rooted triple displayed by $\varphi_{Ad}(\PP|_Y)$, then for every $j$, the restricted tree $T_j|_Y$ contains a maximal cluster $C_j$ that
contains $a$ and $b$, and for a least one of these clusters, $c$ is not also
present
(if $a,b,c$ are contained in $C_j$ for all $j$, then we never get $ab|c$ by the
level 2 condition, but if $a, b$ are in separate maximal clusters 
for some tree in $\PP|_Y$ we get the unresolved triple in $\varphi_{Ad}(\PP|_Y)$).

Using the fact that each tree in $\PP$ has height at most 2, 
it follows that $a$ and $b$ appear together in a maximal cluster of each tree
in $\PP$, and at least one of these maximal clusters does not contain $c$,
 so $ab|c$ is displayed by $\varphi_{Ad}(\PP)$.

For Aho consensus, the maximal clusters of $\varphi_{Ah}(\PP|_Y)$ are the connected components of the BUILD graph $(Y, E_{\PP|_Y})$, and the maximal clusters of  $\varphi_{Ah}(\PP)$  are the connected components of the BUILD graph $(X, E_\PP)$. Since $Y \subset X$ and $E_{\PP|_Y}\subseteq E_\PP$, a maximal cluster of $\varphi_{Ah}(\PP|_Y)$ is the intersection of a maximal cluster of $\varphi_{Ah}(\PP)$ with $Y$, and so is present as a cluster in $\varphi_{Ah}(\PP)|_Y$.  Moreover, the BUILD graph $(W, E_{\PP|_Y})$ for any
maximal cluster $W$ of $\varphi_{Ah}(\PP|_Y)$ is disconnected (by the height 2 condition),  so every cluster of $\varphi_{Ah}(\PP|_Y)$ is present in $\varphi(\PP)|_Y$, as required.
\end{proof}
We note that the consensus problem on level 2 trees has close links with the consensus problem on partitions, itself an area of active  research (see for example, \cite{gue11}).

\subsection{Regular consensus methods that satisfy restricted forms of extension stability.}

The following condition is a special case of extension stability, since it applies only when the consensus of the restricted profile
is binary (i.e. fully resolved).

\bigskip

\begin{enumerate}
	\item[] {\em Binary extension stability}. For every phylogenetic profile $\PP$ consisting of trees from a leaf set $X$ of arbitrarily large size, and  any nonempty set $Y\subseteq X$, if $\varphi(\PP|_Y)$ is a binary tree, then it equals $\varphi(\PP)|_Y.$
\end{enumerate}

\bigskip

A second restricted form of extension stability is that obtained by restricting the choice of $Y$ to either the clusters of 
$\varphi(\PP)$ or to clusters present in all trees in $\PP$.  More formally, we can define the following two properties:

\begin{enumerate}
	\item[] {\em Extension stability for clusters (I)}. For every phylogenetic profile $\PP$ consisting of trees from a leaf set $X$ of arbitrarily large size, and  any nonempty set $Y\subseteq X$, if $Y$ is a cluster of $\varphi(\PP)$, then $\varphi(\PP |_Y) \preceq \varphi(\PP)|_Y.$
\end{enumerate}

\begin{enumerate}
	\item[] {\em Extension stability  for clusters (II)}. For every phylogenetic profile $\PP$ consisting of trees from a leaf set $X$ of arbitrarily large size, and  any nonempty set $Y\subseteq X$, if $Y$ is a cluster of every tree in $\PP$  then $\varphi(\PP |_Y) \preceq \varphi(\PP)|_Y.$
\end{enumerate}

It turns out there are regular consensus methods that satisfy all these restricted versions of extension stability, and, once again,
Adams consensus suffices. The proof of the following result is provided in the Appendix.

\begin{thm}
\label{ad2}
Adams consensus satisfies binary extension stability, extension stability for clusters (I) and extension stability for clusters (II).
\end{thm}

\section{No natural consensus method is extension stable in general}

We turn now to our  main  result, which is the following Arrow-type theorem.
\begin{thm}
\label{mainthm}
No regular consensus method $\varphi$ is extension stable on profiles of two trees.
\end{thm}

In order to prove this theorem, we first require a lemma.

\begin{lem}\label{bryantlem}
Suppose that $\varphi$ is regular consensus method which is extension stable and let $\PP$ be the profile of the two trees $T_1, T_2$ in Fig.~\ref{f:bryfig}.  In that case, $abc$ is a cluster of $\varphi(\PP)$.
\end{lem}

\begin{proof} Suppose that $\varphi(\PP)$ does not contain the cluster $abc$ (we will derive a contradiction).
Restricting the profile to the leaves $\{a,b,d\}$, we have identical trees $ab|d$, so by unanamity $\varphi(\PP|_{\{a,b,d\}}) = ab|d$. Thus $\varphi(\PP)$ is one of the trees $T_A,T_B,T_C$ with $c$ placed as shown in Fig.~\ref{f:bryfig}.  Let $\PP'$ be the profile $T_1',T_2'$ obtained by mapping $(a,b,c,d)$ to $(c,b,a,e)$ and reversing the order of the trees. If $\varphi(\PP)$ equals $T_A$, $T_B$ or $T_C$ then, by unanimity (and anonymity), $\varphi(\PP')$ must equal $T_A'$, $T_B'$ or $T_C'$ respectively. In all three cases, $\varphi(\PP)$ and  $\varphi(\PP')$ resolve the triple $a,b,c$ differently, so the two consensus trees are incompatible.

Now consider the profile $\PP^{+}$ in Fig.~\ref{f:bryfig}, and observe that $\PP = \mathcal{P^+}|_{\{a,b,c,d\}}$ and  $\mathcal{P'} = \mathcal{P^+}|_{\{a,b,c,e\}}$. 
We have assumed that $\varphi$ is extension stable so that 
$\varphi(\PP^+)|_{\{a,b,c,d\}}$ refines $\varphi(\PP^+|_{\{a,b,c,d\}}) = \varphi(\PP)$, and
$\varphi(\PP^+)|_{\{a,b,c,e\}}$   refines $\varphi(\PP^+|_{\{a,b,c,e\}}) = \varphi(\PP')$.
This is  a contradiction, as $\varphi(\PP)$ and $\varphi(\PP')$ are incompatible.

\begin{figure}[ht]
\center
\includegraphics[width=8cm]{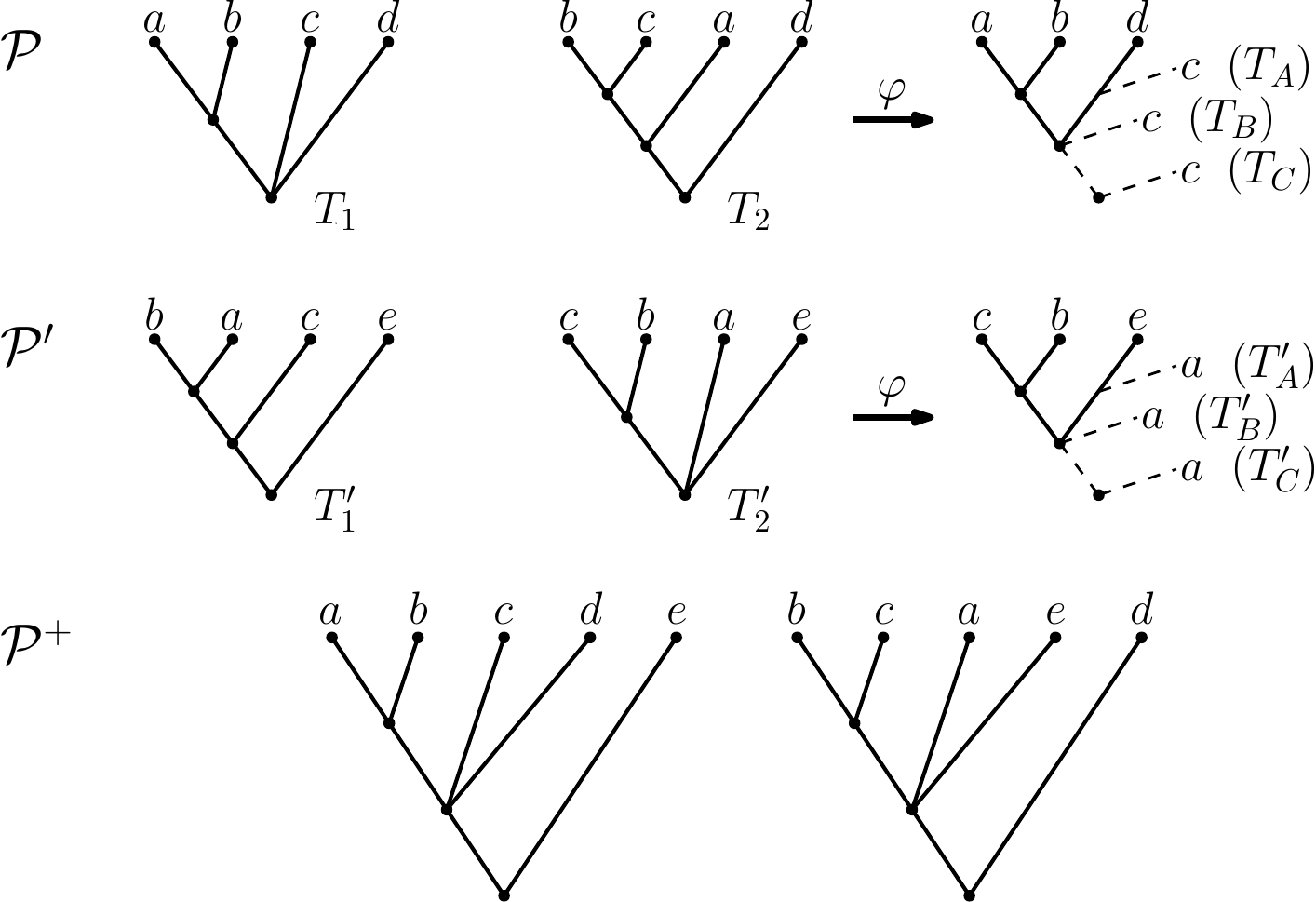}
\caption{Trees required for the proof of Lemma~\ref{bryantlem}.}\label{f:bryfig}
\end{figure}

This contradiction was based on the assumption that $\varphi(\PP)$ does not contain the cluster $abc$, thus $\varphi(\PP)$ contains this cluster.
\end{proof}

\begin{proof}[Proof of Theorem~\ref{mainthm}]
Consider the profile $\PP$ consisting of the two trees shown in Fig.~\ref{f:theseus},  suppose that
$\varphi(\PP)$ is a regular consensus method that satisfies extension stability, and let $T=\varphi(\PP)$ (we will derive a contradiction).

\begin{figure}[htb]
\center
\includegraphics[width=8.5cm]{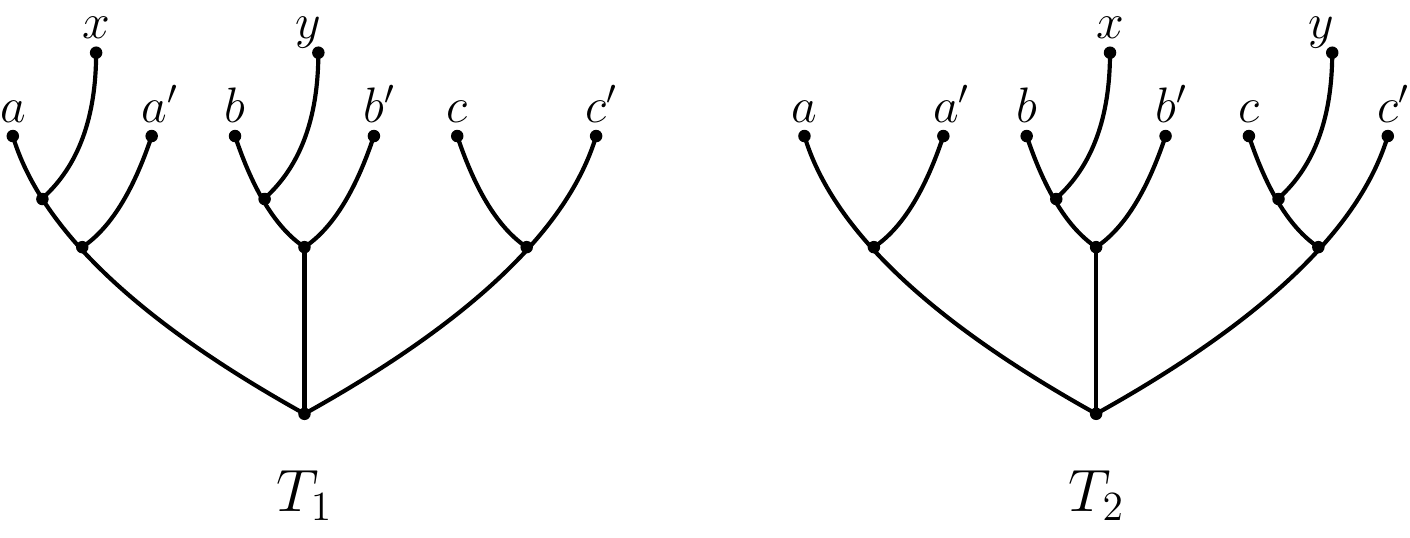}
\caption{The profile $\PP$ of two trees used in the proof by contradiction argument for establishing Theorem~\ref{mainthm}.}\label{f:theseus}
\end{figure}


We apply Lemma~\ref{bryantlem}
repeatedly.  For the leaf set $\{a,a',c, x\}$, we deduce that $\varphi(\PP|_{\{a,a', c,x\}})$ has the nontrivial cluster $aa'x$ and
so displays $c|ax$, while for
the leaf set $\{b,b',c,x\}$, the tree $\varphi(\PP|_{\{b,b', c,x\}})$ has the nontrivial cluster $bb'x$ and so displays $c|bx$.
 Consequently, by extension stability,  $c|ax$ and $c|bx$ must both be displayed by $T$, and therefore $c|ab$ is also displayed by $T$.

Repeating this argument for $y$, with the leaf set $\{a,b,b', y\}$, we deduce that  $\varphi(\PP|_{\{a,a', c,x\}})$ has the  nontrivial cluster $ybb'$
and so displays $a|by$, while for the leaf set $\{a,c,c',y\}$, the tree $\varphi(\PP|_{\{a,c,c', y\}})$ has the nontrivial cluster $ycc'$ and so displays 
$a|cy$. 
Consequently, by extension stability,  $a|by$ and $a|cy$ must both be displayed by $T$, and therefore $a|bc$ is also displayed by $T$.  However, this contradicts
the conclusion of the previous paragraph that $c|ab$ is displayed by $T$, and this contradiction implies that such a consensus method $\varphi$ exists cannot exist. 
\end{proof}


{\bf Remarks:}  The example used in the proofs of Theorem~\ref{mainthm} and Lemma~\ref{bryantlem} shows that the height bound  in Theorem~\ref{posthm} is tight.  That is, although there are regular consensus methods that are extension stable on an arbitrary number of trees of height at most 2, when we move to height 3,  even for just two trees,  extension stability can be impossible to achieve for any consensus method.

\section{Stability when trees rather than taxa are added}

\label{secondstab}
So far, the stability properties we have considered (extension stability and its variations/restrictions) involve adding taxa to extend a given profile of trees.   In this final section, we consider the taxa as being fixed and investigate a stability condition that applies as phylogenies are added to the profile.

Suppose a consensus tree $T$ has been built  from a sequence of input trees $\PP=(T_1, T_2, \ldots, T_{k-1})$ and then a new input tree $T_{k}$ becomes available.   A natural question now arises: Is the resulting consensus of these $k$ trees the same as simply taking the consensus of the most current consensus tree $T$ and the new tree $T_{k}$?  
For certain methods (e.g. strict consensus), this applies, whereas for others (e.g. majority rule), it can fail.

More formally, consider the following condition:
For all $k \geq 1$,  and  every profile $\PP = (T_1, \ldots, T_k)    \in RP(X)^k$, we have:     
\begin{equation}
\label{ass}
\varphi(\varphi(\PP - T_k), T_k) = \varphi(\PP),
\end{equation}
where $\PP-T_k$ is the profile obtained by removing tree $T_k$ from the last position in the profile to give a profile of length $(k-1)$. 
If $\varphi$ satisfies Property (\ref{ass}),  we say that it is {\em associatively stable}. 
Associative stability is computationally desirable, since all the information required to update a consensus tree when a new input tree arises is the current consensus tree, not the list of previous input trees used to produce it.

Notice that any consensus method $\varphi$ on $RP(X)$ can be regarded as a binary operation $\circ$ on this set, by
letting $T \circ T' := \varphi(T, T')$. If $\varphi$ is a regular consensus method, this operation is both commutative and (by unanimity) it is also idempotent (i.e. $T \circ T=T$) and, if  $\varphi$ is also associatively stable, then the operation $\circ$ is associative (i.e. $(T \circ T') \circ T'' = T \circ (T' \circ T'')$).    Note that the associativity of $\circ$ is a weaker condition than associative stability.

Moreover, any regular consensus method that is associatively stable  has the property that for 
any profile $\PP  \in RP(X)^k$ the following product
$$
T_1 \circ T_2 \circ \cdots \circ T_k,
$$
is well-defined, takes the same value for any ordering $T_1, T_2, \ldots T_k$ of the trees in $\PP$, and can be computed by repeatedly taking pairwise consensus operations on  trees (placing the brackets arbitrarily).

The proof of the following result is provided in the Appendix.
\begin{lem}
\label{natlem}
If $\varphi$ is a regular consensus method that is associatively stable then $\varphi(\PP)$ depends only on the set of trees present in $\PP$ and not on their frequency.
\end{lem}

It follows that consensus methods like majority rule are not associatively stable. 

 Although strict consensus satisfies associative stability, loose consensus fails to have this property. A simple example is provided by
 the profile $\PP = (ab|c, ac|b, bc|a)$ for which $bc$ is a cluster of $\varphi(\varphi((ab|c, ac|b)), bc|a)$ but $bc$ is not a cluster of
 $\varphi(\PP)$

\begin{prop}
Adams consensus is associatively stable for profiles of trees having height 2, but not for profiles of trees of height four. Aho consensus fails to be associatively stable even for profiles of trees having height  2. 
\end{prop}

\begin{proof}
Adams consensus is associatively stable for any profile $\PP$ of trees of height 2, since the nontrivial clusters of $\varphi_{Ad}(\PP)$ consist of the nonempty intersections of the nontrivial clusters of the trees in $\PP$, and intersection is an associative operation on sets. 

However, Adams consensus is not associative (and so not associatively stable) in general, by virtue of the example provided in Fig.~\ref{f:both_fail}, which involves three trees of height 4. 

Aho consensus also fails to be associative for the profile in Fig.~\ref{f:both_fail}, but a simpler example shows that this method also fails associativity on trees of height 2.  Consider the profile   $\PP = (T_1=((abcd), (ef)), T_2=((cdef), (ab)), T_3=((abef), (cd)))$ consisting of  three trees of height 2.  We have
$\varphi_{Ah}(T_1, T_2) = ((ab), (ef), c, d)$ and so $\varphi_{Ah}(\varphi_{Ah}(T_1, T_2), T_3) = ((ab), (ef), c,d)$; however as we noted earlier, $\varphi_{Ah}(\PP)$ is the star tree. 
Thus Aho consensus fails to be associative (and thereby associatively stable) on trees of height 2.
\end{proof}

\mbox{ }

\begin{figure}[ht]
\center
\includegraphics[width=8cm]{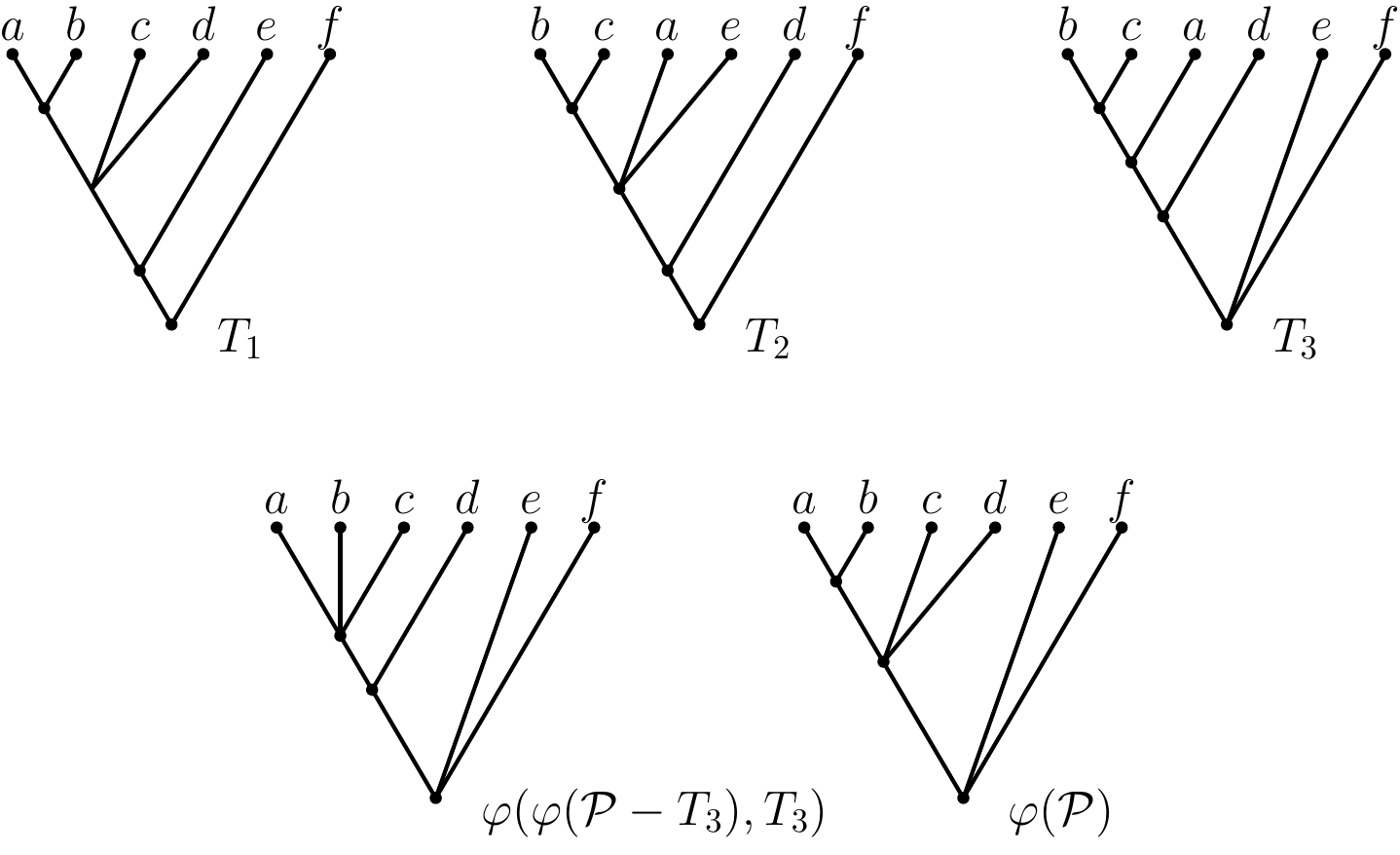}
\caption{A profile of three trees on which both Adams consensus (and also Aho consensus) fail to be associative (and hence are not associatively stable).}
\label{f:both_fail}
\end{figure}

\section{Conclusing comments}

We have explored two notions of stability for consensus methods:  one where taxa are added, and the other where trees are added. In the former setting we have answered a question left open by \cite{15vel}: not only are no current consensus methods extension stable in general, but it is impossible to design one that is while retaining the usual (`regular') properties expected of any consensus method (Theorem~\ref{mainthm}).    In other words, whichever method is used to combine trees, it may still be necessary to revise phylogenetic relationships in future, even when the new trees on the taxa they originally considered agree with the original trees on which those relationships were established.  However, if the new taxa lie outside a cluster (monophyletic clade) that  appears in the consensus tree, then existing relationships can be preserved (Theorem~\ref{ad2}), at least for certain consensus methods.

We also investigated the stability of consensus methods as trees, rather than taxa, are added. Associative stability is the requirement that the current consensus tree provides a sufficient summary of the earlier input trees when we wish to update our phylogeny when presented with a new input tree,  rather than having to recompute the consensus afresh using all the input trees. Certain methods (e.g. strict consensus) satisfy associative stability, however we showed that certain other methods fail to satisfy it except in special cases.  

Because our results in these sections are stated quite generally, they apply to the performance of consensus methods across a wide range of applications in systematic biology, including inferring species trees from gene trees \cite{7deg09}, summarizing the posterior distribution  of trees in Bayesian phylogenetics \cite{10hol08}, and inferring a species tree from a collection of species trees obtained from different studies. 

Our results also raise a number of interesting questions for further work which we now discuss.      

   We saw  from Theorem~\ref{posthm} that Adams and Aho consensus are extension stable for trees of height at most 2. However, Adams is not extension stable for trees of greater height, even when the trees are binary.  An example is the profile in Figure~\ref{f:adams.binary}.  
For this example, the Aho consensus method does not violate extension stability; however, it is easy to find two non-binary trees and a subset $Y$ for which it does
 (e.g. the profile $\PP^+$ in Fig.~\ref{f:bryfig}, which consists of the trees $(((ab), c,d), e)$ and $(((bc), a,e), d)$ and $Y=\{a,b,c,d\}$). Moreover, by resolving the polytomies in each tree in two different ways,  we obtain a profile of four binary trees on which Aho consensus violates extension stability for the same set $Y$.

\begin{figure}[ht]
\center
\includegraphics[width=8cm]{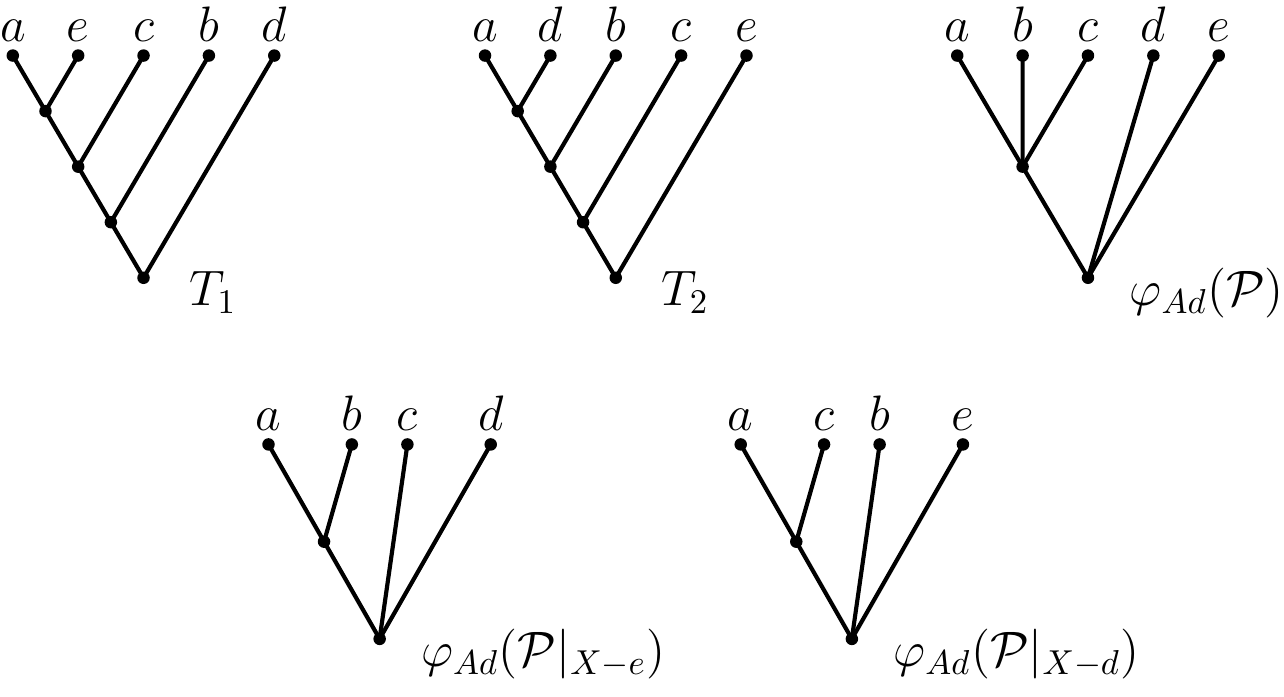}
\caption{Profile of two binary trees (top left) on which  Adams consensus  fails to be extension stable. The Adams consensus for this profile (top right) has one nontrivial cluster ($abc$) whereas the Adams consensus trees restricted to two subsets (bottom) display the rooted trees $ab|c$ and $b|ac$ respectively (and so are not compatible with any tree). }
\label{f:adams.binary}
\end{figure}

These observations, together with the results in the previous sections, suggest our first  question:

\begin{itemize}
\item 
 Is there a regular consensus method that  satisfies extension stability when the method is restricted to profiles of binary trees? Or even just to profiles of {\em two} binary trees?
  \end{itemize}

Regarding associative stability, we have seen that although some methods, such as strict consensus, satisfy this property, other methods (including loose consensus, majority consensus, Adams and Aho consensus) fail to be associative.  This raises our second interesting  open question. 

\begin{itemize}
\item 
Is there a regular consensus method that satisfies associative stability and is Pareto on rooted triples?  
\end{itemize}

  In addition to these questions, further topics that may be useful to consider are extensions of the concepts and results we have described  to allow profiles of trees in which each tree comes with given branch lengths and/or a weight (e.g. a posterior probability assignment).  One step in this direction has been provided in a recent study (\cite{vu16}) into the consistency and stability of on-line updating of Bayesian posterior distributions on trees as new data is obtained.

 Finally, certain consensus methods not discussed here, such as  greedy consensus \cite{7deg09} may involve a randomisation step, such as breaking a tie arbitrarily.  Strictly speaking, such methods are not  consensus methods (i.e. functions) as we have defined them, however all the axioms we have studied could be rephrased as properties that apply when a consensus tree is a random variable whose distribution is determined by the input trees.



\section*{Funding}
We thank the University of Canterbury Erskine Fund for financial support for this project.  

\section*{Acknowledgments}


We thank F. R. McMorris and Peter Lockhart for advice regarding background literature.  
We also thank F. R. McMorris and a second (anonymous) reviewer for numerous helpful suggestions.

\section{Appendix}

\begin{lem}
The condition $\varphi(\PP \cdot \varphi(\PP)) = \varphi(\PP)$ holds for strict consensus, loose consensus, majority rule, Adams consensus and Aho consensus.
\end{lem}

\begin{proof}
It is straightforward to verify that strict consensus and loose consensus satisfy this property. 
For majority rule, suppose that $A$ is a cluster in $\varphi(\PP)$. Then $A$ is a cluster in more than half the trees in $\PP$, and since $A$ is a cluster of $\varphi(\PP)$ also,
$A$ is a cluster in  more than half the trees of $\PP \cdot \varphi(\PP)$. Thus, $A$ is a cluster of $\varphi(\PP\cdot \varphi(\PP))$.  Conversely, if
$A$ is not a cluster of $\varphi(\PP)$, then  $A$ is a cluster in  at most half the trees in $\PP$, and since $A$ is therefore not in $\varphi(\PP)$, $A$ is a cluster
in less than half the trees of $\PP \cdot \varphi(\PP)$. Thus $A$ is not a cluster of  $\varphi(\PP\cdot \varphi(\PP))$.


For Aho consensus, observe that the rooted triples displayed by $\varphi_{Ah}(\PP)$ contain the set $R$ of rooted triples that are displayed by every tree in
$\PP$, the set of rooted triples displayed by every tree in $\PP \cdot \varphi_{Ah}(\PP)$ equals $R$ and therefore:
$$\varphi_{Ah}(\PP\cdot \varphi_{Ah}(\PP))=\varphi_{Ah}(\PP).$$

The proof of Adams consensus follows from the characterisation of the Adams consensus tree in terms of nestings provided by Theorems 2 and 3 of \cite{1ada86}.

\end{proof}

We now prove Theorem~\ref{ad2}.

\begin{proof}
We establish the binary extension stability property first. 
The proof relies on the fact that Adams consensus satisfies a type of converse
to unanimity when the trees are binary. Namely, if the Adams tree is a binary tree then every input tree has to be identical to this tree. Formally, 
\begin{equation}
\label{adamsbin}
T \mbox{ binary and } \varphi(\PP)= T \mbox { implies }  \PP = (T, T, \ldots, T).
\end{equation}
To see that this holds for $\varphi= \varphi_{Ad}$, suppose that $\varphi_{Ad}(\PP)$ is a binary tree.  Then (i) all of the trees in $\PP$
must have exactly two maximal proper clusters, and (ii) the pair of maximal proper clusters of each tree in $\PP$ is the same for each tree in $\PP$.  
Property (i) holds because if one tree had three maximal proper clusters, then, by construction, $\varphi_{Ad}(\PP)$ would contain at least three maximal proper clusters and so would not be binary.  For Property (ii), observe that if there were two binary trees in $\PP$ that
had a different pair of maximal clusters, say $\{A, \overline{A}\}$ and $\{B, \overline{B}\}$, then at least
three nonempty intersections exist in the set $\{A \cap B, \overline{A} \cap B, A \cap \overline{B}, \overline{A} \cap \overline{B}\}$, and so, once again, $\varphi_{Ad}(\PP)$ would have at least three maximal proper clusters and so would  not be binary.  
Thus, the maximal proper clusters of $\varphi_{Ad}(\PP)$ agree with the maximal proper clusters of each tree in $\PP$.  The recursive nature of the Adams consensus  construction now means that this argument can be repeated at subsequent levels, to ensure that the trees in $\PP$ are all binary and identical.

It follows that if $\varphi_{Ad}(\PP|_Y)$ is a binary tree $T$ on $Y$, then $\PP|_Y = (T, \ldots, T)$, and so each rooted triple displayed by $T$ is also displayed by $\varphi_{Ad}(\PP)$, and hence $T$ is displayed by
$\varphi_{Ad}(\PP)$. Now, since $T$ is binary, this means that $\varphi_{Ad}(\PP)|_Y =T$, as claimed.

To  establish extension stability for clusters (I), notice that if $Y$ is a cluster of   $\varphi_{Ad}(\PP)$, then, by the recursive nature of the Adams consensus  method, $\varphi_{Ad}(\PP)|_Y = \varphi_{Ad}(\PP|_Y)$, which is an even stronger condition than extension stability for clusters (I).  For extension stability for clusters (II), if $Y$ is a cluster of each tree in $\PP$, then $Y$ is a cluster of
$\varphi_{Ad}(\PP)$, and so (again  by the recursive nature of the Adams consensus method) we have $\varphi_{Ad}(\PP)|_Y = \varphi_{Ad}(\PP|_Y)$, which establishes the property. 
\end{proof}

Finally, we prove Lemma~\ref{natlem}.

\begin{proof}
For $\PP = (T_1, T_2, \ldots, T_k)$, let $T'_1, \ldots, T'_r$ be a  list of the distinct trees that appear at least once in $\PP$ (so that $r \leq k$).   Thus $\PP = (T'_{f(1)}, \ldots, T'_{f(k)})$ for some surjection $f: \{1, \ldots, k\} \rightarrow \{1, \ldots, r\}$.  We will show that $\varphi(\PP) = T'_1 \circ T'_2 \circ \cdots \circ T'_r$, which establishes the lemma.
Since $\varphi$ is associatively stable, we can write 
$\varphi(\PP) = T_1 \circ T_2 \circ \cdots \circ T_k$, and since $\varphi$ satisfies commutatively it follows that 
$\varphi(\PP) = (T'_1)^{n_1} \circ (T'_2)^{n_2} \circ \cdots \circ (T'_r)^{n_r}$,
where $n_i$ is the number of times that tree $T'_i$ appears in $\PP$, and $(T'_i)^{n_i} = T'_i \circ \cdots  \circ T'_i$ [$n_i$ times].
Unanimity now ensures that $(T'_i)^{n_i} = T'_i$ for all $i$,  so $\varphi(\PP) = T'_1 \circ T'_2 \circ \cdots \circ T'_r$ as claimed. 
\end{proof}


\end{document}